\documentclass[11pt]{article}

\usepackage{amsmath}
\usepackage{amsfonts}
\usepackage{graphicx}
\usepackage{amssymb}
\usepackage{multirow}
\usepackage{geometry}
\usepackage{soul}
\usepackage{colortbl}
\usepackage{wrapfig}
\usepackage{algorithm}
\usepackage{algorithmicx}
\usepackage{algpseudocode}
\usepackage{todonotes}
\usepackage{amsthm}
\usepackage{thmtools} 
\usepackage{thm-restate}

\usepackage{xcolor}
\usepackage{soul}

\algtext*{EndWhile}
\algtext*{EndIf}
\algtext*{EndFor}

\newtheorem{lemma}{Lemma}
\numberwithin{lemma}{section}
\newtheorem{theorem}[lemma]{Theorem}

\newtheorem{observation}[lemma]{Observation}
\newtheorem{fact}[lemma]{Fact}

\newtheorem{claim}[lemma]{Claim}

\usepackage{wrapfig}
\usepackage{bm}

\usepackage{booktabs}
\usepackage[pdftex, plainpages = false, pdfpagelabels, 
                 bookmarks=false,
                 bookmarksopen = true,
                 bookmarksnumbered = true,
                 breaklinks = true,
                 linktocpage,
                 pagebackref,
                 colorlinks = true,  
                 linkcolor = blue,
                 urlcolor  = blue,
                 citecolor = red,
                 anchorcolor = green,
                 hyperindex = true,
                 hyperfigures
                 ]{hyperref}

\usepackage{wrapfig}
\usepackage{bm}

\usepackage{graphicx,color}
\usepackage{empheq}
\usepackage{tcolorbox}
\definecolor{blueish}{rgb}{0.122, 0.435, 0.698}
\definecolor{dagstuhlyellow}{rgb}{0.99,0.78,0.07}
\definecolor{lightgray}{rgb}{0.9,0.9,0.9}
\newtcbox{\colbox}{
size=title,
  nobeforeafter,
  colframe=white,
  colback=blue!5!white,
  arc=10pt,
  tcbox raise base}

\usepackage{xspace}

\title{Bipartizing (Pseudo-)Disk Graphs: \\ Approximation with a Ratio Better than 3}

\author{
Daniel Lokshtanov\footnote{University of California, Santa Barbara, USA. Email: \texttt{daniello@ucsb.edu}} \and Fahad Panolan\footnote{University of Leeds, UK. Email: \texttt{F.panolan@leeds.ac.uk}} \and Saket Saurabh\footnote{Institute of Mathematical Sciences, Chennai, India. Email: \texttt{saket@imsc.res.in}} \and Jie Xue\footnote{New York University Shanghai, China. Email: \texttt{jiexue@nyu.edu}} \and Meirav Zehavi\footnote{Ben-Gurion University, Israel. Email: \texttt{meiravze@bgu.ac.il}}}

\date{}

\begin{document}

\maketitle

\begin{abstract}
    In a disk graph, every vertex corresponds to a disk in $\mathbb{R}^2$ and two vertices are connected by an edge whenever the two corresponding disks intersect.
    Disk graphs form an important class of geometric intersection graphs, which generalizes both planar graphs and unit-disk graphs.
    We study a fundamental optimization problem in algorithmic graph theory, \textsc{Bipartization} (also known as {\sc Odd Cycle Transversal}), on the class of disk graphs.
    The goal of \textsc{Bipartization} is to delete a minimum number of vertices from the input graph such that the resulting graph is bipartite.
    A folklore (polynomial-time) $3$-approximation algorithm for \textsc{Bipartization} on disk graphs follows from the classical framework of Goemans and Williamson [Combinatorica'98] for cycle-hitting problems.
    For over two decades, this result has remained the best known approximation for the problem (in fact, even for \textsc{Bipartization} on unit-disk graphs).
    In this paper, we achieve the first improvement upon this result, by giving a $(3-\alpha)$-approximation algorithm for \textsc{Bipartization} on disk graphs, for some constant $\alpha>0$.
    Our algorithm directly generalizes to the broader class of \textit{pseudo-disk} graphs.
    Furthermore, our algorithm is \textit{robust} in the sense that it does not require a geometric realization of the input graph to be given.
\end{abstract}

\section{Introduction}
Disk graphs refer to intersection graphs of disks in the plane $\mathbb{R}^2$.
Formally, in a disk graph, every vertex corresponds to a disk in $\mathbb{R}^2$ and two vertices are connected by an edge whenever the two corresponding disks intersect.
As a rather general class of geometric intersection graphs, disk graphs simultaneously generalize two important graph classes, unit-disk graphs and planar graphs, both of which have been extensively studied over decades.
Many central problems in algorithmic graph theory have been considered on disk graphs, including \textsc{Vertex Cover} \cite{ErlebachJS05,lokshtanov2023framework,Leeuwen06}, \textsc{Independent Set}~\cite{ErlebachJS05}, \textsc{Maximum Clique}~\cite{bonamy2021eptas}, \textsc{Feedback Vertex Set}~\cite{lokshtanov2022subexponential,lokshtanov2023framework}, \textsc{Dominating Set}~\cite{GibsonP10}, etc.

In this paper, we investigate a fundamental optimization problem on the class of disk graphs, called {\sc Bipartization}.
In this problem, the input is a graph $G$ and the goal is to delete a smallest number of vertices from $G$ such that the resulting graph is bipartite.
There has been a long line of work studying {\sc Bipartization} (e.g., see~\cite{FioriniHRV07,DBLP:conf/soda/KawarabayashiR10,KawarabayashiR10,RautenbachR01,Reed99,ReedSV04} and citations therein).
Observe that the edge counterpart of the  {\sc Bipartization} problem,  {\sc Edge Bipartization}, where we need to find fewest edges whose deletion results in a bipartite graph, is equivalent to the classical {\sc Maximum Cut} problem (which has been studied for over five decades~\cite{DBLP:journals/jacm/GoemansW95,DBLP:journals/jacm/SahniG76}). 
It is known that {\sc Edge Bipartization} (and {\sc Maximum Cut}) reduces to {\sc Bipartization}~\cite{wernicke2014algorithmic}, and both problems find applications in computational biology~\cite{wernicke2014algorithmic,RizziBIL02}, VLSI chip design~\cite{KahngVZ01}, genome sequence assembly~\cite{pop2004hierarchical}, and more.
{\sc Bipartization} is of particular interest also due to the following observation: a graph is bipartite if and only if it does not contain any odd cycle.
As such, it can be formulated as hitting all odd cycles in the graph (using fewest vertices).
For this reason, the {\sc Bipartization} problem is also known as {\sc Odd Cycle Transversal}, and belongs to the family of \textit{cycle-hitting} problems, one of the most well-studied topics in algorithmic graph theory.
Besides \textsc{Bipartization}, other important cycle-hitting problems that have been studied on disk graphs include \textsc{Feedback Vertex Set}, \textsc{Triangle Hitting}, \textsc{Short Cycle Hitting}, etc.

There is no surprise that {\sc Bipartization} is NP-complete.
In fact, it is NP-complete even on graphs of maximum degree $3$ and planar graphs of maximum degree $4$~\cite{ChoiNR89}. 
As such, the study of the problem is mainly in the context of approximation algorithms and parameterized algorithms.
On the parameterized front, it was known that {\sc Bipartization} can be solved in $2^{O(k)} \cdot n^{O(1)}$ time where $k$ is the solution size \cite{KawarabayashiR10,LokshtanovNRRS14,reed2004finding}.
On planar graphs and unit-disk graphs, there exist improved algorithms with running time $k^{O(\sqrt{k})} \cdot n^{O(1)}$ \cite{bandyapadhyay2022subexponential,bandyapadhyay2022true,lokshtanov2012subexponential,marx2022framework}, which are almost tight assuming the Exponential-Time Hypothesis (ETH).
On disk graphs, it was not known whether {\sc Bipartization} admits a subexponential-time parameterized algorithm, until very recently Lokshtanov et al. gave a $k^{O(k^{27/28})} \cdot n^{O(1)}$-time algorithm~\cite{lokshtanov2022subexponential}.

From the perspective of approximation algorithms (which is the focus of this paper), {\sc Bipartization} is one of the trickiest problems in the sense that no polynomial-time approximation scheme (PTAS) was known on any (nontrivial) graph classes, but at the same time no inapproximability results was known except for general graphs.
On general graphs, {\sc Bipartization} cannot admit any (polynomial-time) constant-approximation algorithm assuming the Unique Games Conjecture~\cite{bansal2009optimal}, and the best known approximation ratio is $O(\sqrt{\log {\sf opt}})$ due to Kratsch and Wahlstr\"{o}m~\cite{KratschW20}, improving the earlier bounds of $O(\log n)$~\cite{EvenNRS00} and $O(\sqrt{\log n})$~\cite{AgarwalCMM05}.
It has been a long-standing open question whether {\sc Bipartization} admits a PTAS on planar graphs or unit-disk graphs.
On planar graphs, the currently best approximation for {\sc Bipartization} is still the $\frac{9}{4}$-approximation algorithm given in the seminal work of Goemans and Williamson~\cite{GoemansW98} more than two decades ago.
This result immediately gives a $3$-approximation algorithm for {\sc Bipartization} on disk graphs (and in particular, unit-disk graphs) by the well-known fact that triangle-free disk graphs are planar \cite{kratochvil1996intersection}.

\begin{theorem}[folklore] \label{thm-3appx}
There exists a polynomial-time $3$-approximation algorithm for \textsc{Bipartization} on the class of disk graphs.
\end{theorem}
\begin{proof}
Let $G$ be the input disk graph.
We repeat the following step until $G$ contains no triangles: find a triangle in $G$, add its three vertices to the solution, and remove them from $G$.
Denote by $S$ the set of vertices added to the solution and by $G' = G-S$ the resulting triangle-free graph, which is planar \cite{kratochvil1996intersection}.
Now apply the algorithm of Goemans and Williamson~\cite{GoemansW98} on $G'$ to obtain a $\frac{9}{4}$-approximation solution $S'$.
Note that $S \cup S'$ is a $3$-approximation solution (for $G$).
Indeed, any solution of \textsc{Bipartization} must hit all triangles in $G$ and thus contains at least $|S|/3$ vertices in $S$.
Also, it contains at least $|S'|/\frac{9}{4}$ vertices in $V(G')$.
So its size is at least $|S|/3 + |S'|/\frac{9}{4} \geq |S \cup S'|/3$.
\end{proof}

For over two decades, this has remained the best approximation algorithm for \textsc{Bipartization} on disk graphs (in fact, even on unit-disk graphs).
Thus, there is a natural question to be asked: can we achieve an approximation ratio better than $3$ for the problem?
Note that one cannot achieve this by improving the approximation ratio $\frac{9}{4}$ for \textsc{Bipartization} on planar graphs.
Indeed, even if we had a PTAS on planar graphs, the above argument still only gives us a $3$-approximation algorithm on disk graphs.
Therefore, the number $3$ here is truly a bottleneck of the approximation ratio of the problem.
In this paper, we break this bottleneck and answer the above question affirmatively by giving the first algorithm for \textsc{Bipartization} on disk graphs with an approximation ratio better than 3.
Specifically, our main result is the following. 

\begin{restatable}{theorem}{main}\label{thm:Main}
There exists a polynomial-time $(3-\alpha)$-approximation algorithm for \textsc{Bipartization} on the class of disk graphs, for some constant $\alpha > 0$.
\end{restatable}


We remark that Theorem~\ref{thm:Main} should be viewed as a  proof of concept, rather than the quantitative improvement. 
Our algorithm in Theorem~\ref{thm:Main} is \textit{robust} in the sense that it does not require the geometric realization of the input disk graph to be given.
Furthermore, our algorithm directly generalizes to the broader class of \textit{pseudo-disk} graphs, which are the intersection graphs of topological disks in which the boundaries of every two of them intersect at most twice.
(Note that the 3-approximation algorithm in Theorem~\ref{thm-3appx} also works for pseudo-disk graphs as triangle-free pseudo-disk graphs are planar \cite{kratochvil1996intersection}.)
Again, this generalized algorithm is robust.

\begin{restatable}{theorem}{mainpseudo}\label{thm:Main2}
There exists a polynomial-time $(3-\alpha)$-approximation algorithm for \textsc{Bipartization} on the class of pseudo-disk graphs, for some constant $\alpha > 0$.
\end{restatable}

Finally, we observe that our technique not only works for the \textsc{Bipartization} problem.
In fact, it can be applied to a large category of vertex-deletion problems on (pseudo-)disk graphs, resulting in $(3-\alpha)$-approximation algorithms.
Although at this point it only yields interesting results for \textsc{Bipartization} (mainly because most well-studied problems in the category already have algorithms on disk graphs with approximation ratio better than 3), this demonstrates that our technique could possibly have further applications in the future.
We shall briefly discuss this part in Section~\ref{sec-generalization}.



\paragraph{Other related work.}
NP-complete optimization problems on disk graphs and other geometric intersection graphs have received considerable attention over years.
Here we only summarize some recent work on this topic.
The work of de Berg et al.~\cite{DBLP:journals/siamcomp/BergBKMZ20} gave a framework for designing ETH-tight exact algorithms on (unit-)disk graphs or more generally (unit-)ball graphs, which works for a variety of classical optimization problems.
Bonamy et al.~\cite{bonamy2021eptas} presented the first EPTAS and subexponential-time algorithm for \textsc{Maximum Clique} on disk graphs.
Fomin et al.~\cite{FominLPSZ19} designed almost ETH-tight parameterized algorithms for various cycle-packing and cycle-hitting problems on unit-disk graphs; in a follow-up paper~\cite{ZehaviFLP021}, the same authors improved some of their algorithms to be ETH-tight.
Recently, Lokshtanov et al.~\cite{lokshtanov2022subexponential,lokshtanov2023framework} proposed frameworks for subexponential parameterized algorithms and EPTASes for various vertex-deletion problems on disk graphs (the framework for EPTASes does not work for {\sc Bipartization}, while the one for subexponential parameterized algorithms works).
A very recent work by the same authors~\cite{lokshtanov20241} gave a 1.9999-approximation for \textsc{Vertex Cover} on string graphs (i.e., intersection graphs of arbitrary connected geometric objects in the plane), which has the same flavor as this paper.

Besides the aforementioned algorithmic research on {\sc Bipartization}, the problem was also studied in the context of kernelization complexity.
The seminal work by Kratsch and Wahlstr{\"{o}}m~\cite{kratsch2014compression} showed that {\sc Bipartization} admits a randomized polynomial kernel with respect to $k$.
Later, for planar graphs, it was shown to admit a deterministic polynomial kernel by Jansen et al.~\cite{jansen2021deterministic}.
Moreover, the kernelization complexity of {\sc Bipartization} was studied also with respect to some structural parameterizations~\cite{jansen2011polynomial}. 

On a related note, we remark that structural properties of odd cycles in a graph has also received significant attention from various combinatorial points of view.
While the survey of these results is beyond the scope of this paper, as an illustrative example, let us mention the study of Erdős–Pósa properties for odd cycles (see e.g., ~\cite{FioriniHRV07,kawarabayashi2007erdHos,rautenbach2001erdos,thomassen2001erdHos}).

\section{Preliminaries}

Let $G$ be a graph.
We use $V(G)$ and $E(G)$ to denote the vertex set and edge set of $G$, respectively.
For a subset $V \subseteq V(G)$, denote by $G[V]$ the subgraph of $G$ induced by $V$, and by $G - V$ the subgraph of $G$ induced by $V(G) \backslash V$. For a vertex $v$, we use $N_G(v)$ to denote the set  $\{x\in V(G)\setminus \{v\}~\colon~ (x,v)\in E(G)\}$. For a vertex subset $S$, we use $N_G(S)$ and $N_G[S]$ to denote the sets $\bigcup_{z\in S} N_G(z)\setminus S$ and $S\cup N_{G}(S)$, respectively. 
For a vertex $v$ in $G$, we use $d_G(v)$ to denote the degree of $v$ (i.e., $|N_G(v)|$) in $G$. 
A vertex subset $I \subseteq V(G)$ is a {\em distance-$d$ independent} set in $G$, if for any two distinct vertices $x$ and $y$ in $I$, the distance between $x$ and $y$ in $G$ is strictly more than $d$. Here, the distance between two vertices is the number of edges in a shortest path between those vertices.
Let $\mathcal{S}$ be a collection of subsets of $V(G)$.
A \textit{packing} of $\mathcal{S}$ is a sub-collection $\mathcal{S}'$ such that $S \cap S' = \emptyset$ for all $S,S' \in \mathcal{S}'$ with $S \neq S'$.
We say a packing $\mathcal{S}' \subseteq \mathcal{S}$ is \textit{maximal} if any $\mathcal{S}'' \subseteq \mathcal{S}$ satisfying $\mathcal{S}' \subsetneq \mathcal{S}''$ is not a packing of $\mathcal{S}$, and is \textit{maximum} if any subset $\mathcal{S}'' \subseteq \mathcal{S}$ satisfying $|\mathcal{S}'| < |\mathcal{S}''|$ is not a packing of $\mathcal{S}$.
Any maximum packing of $\mathcal{S}$ has the same size.

An \textit{odd cycle transversal} (or \textit{OCT} for short) of $G$ is a subset $S \subseteq V(G)$ such that $G-S$ is a bipartite graph.
A \textit{triangle} in $G$ refers to a set $T = \{u,v,w\}$ of three vertices of $G$ such that $(u,v),(v,w),(w,u) \in E(G)$.
We use $\Delta(G)$ to denote the family of all triangles in $G$.
For a set $\mathcal{T}$ of triangles in $G$, we denote by $V(\mathcal{T})$ the set of vertices of all triangles in $\mathcal{T}$, i.e., $V(\mathcal{T}) = \bigcup_{T \in \mathcal{T}} T$.
The notation $\mathsf{tri}(\mathcal{T})$ denotes the size of a maximum packing of $\mathcal{T}$.

\section{Our algorithm}

In this section, we present our approximation algorithm for \textsc{Bipartization} on disk graphs.
Our algorithm first takes a simple preprocessing step, which reduces the general problem to the problem on $K_4$-free disk graphs (i.e., disk graphs without cliques of size 4).
Then in the main part of our algorithm, we solve \textsc{Bipartization} on $K_4$-free disk graphs.
For a cleaner exposition, we shall first give a randomized version of our algorithm, as it is more intuitive and yields a better approximation ratio.
We then show how to derandomize it in Appendix~\ref{sec-derandom}.

\subsection{Preprocessing: reducing to the $K_4$-free case}

For a graph $G$, let $K_4(G)$ be the set of all $K_4$'s in $G$.
The following lemma allows us to reduce the problem to \textsc{Bipartization} on $K_4$-free disk graphs.

\begin{lemma}
\label{lem:k4reduction}
Let $G$ be a graph and ${\cal C}$ be a packing of $K_4(G)$. Let $S$ be a $\rho$-approximation solution for \textsc{Bipartization} on $G-V({\cal C})$. Then, $S\cup V({\cal C})$ is a $\max\{2,\rho\}$-approximation solution for \textsc{Bipartization} on $G$.
\end{lemma}
\begin{proof}
First, it is clear that $S\cup V({\cal C})$ is an OCT of $G$, since $G - S\cup V({\cal C}) = (G-V({\cal C})) - S$ contains no odd cycles.
Let $S^*$ be an optimal OCT of $G$. 
For every $K_4$ in ${\cal C}$, $S^*$ must contain at least $2$ vertices in the $K_4$ (for otherwise $G-S^*$ contains a triangle).
Since the $K_4$'s in ${\cal C}$ are disjoint, we have $|S^* \cap V(\mathcal{C})| \geq 2|\mathcal{C}| = |V(\mathcal{C})|/2$.
Furthermore, $S^* \cap (V(G) \backslash V(\mathcal{C}))$ is an OCT of $G - V(\mathcal{C})$, which implies $|S^* \cap (V(G) \backslash V(\mathcal{C}))| \geq |S|/\rho$.
Therefore,
\begin{equation*}
    |S^*| = |S^* \cap V(\mathcal{C})| + |S^* \cap (V(G) \backslash V(\mathcal{C}))| \geq \frac{|V(\mathcal{C})|}{2} + \frac{|S|}{\rho} \geq \frac{|S|+|V(\mathcal{C})|}{\max\{2,\rho\}}.
\end{equation*}
As $|S \cup V(\mathcal{C})| = |S|+|V(\mathcal{C})|$, we have $|S \cup V(\mathcal{C})| \leq \max\{2,\rho\} \cdot |S^*|$.
\end{proof}

One reason for why this reduction helps us is the degeneracy of a $K_4$-free disk graph.
Recall that a graph $G$ is $c$-\textit{degenerate} if we can sort its vertices as $v_1,\dots,v_n$ such that each $v_i$ is neighboring to at most $c$ vertices in $\{v_1,\dots,v_{i-1}\}$.
A $c$-degenerate graph of $n$ vertices has at most $cn$ edges and thus has average degree at most $2c$.
We prove that every $K_4$-free disk graph is $11$-degenerate.
To prove this we use the following known result.

\begin{lemma}[\cite{Marathe1995SimpleHF}]
\label{lem:kissingnumber}
Let $D$ be a disk of radius $r$. Let ${\cal S}$ be a set of pairwise non-overlapping  disks of radius $r$ such that every disk in ${\cal S}$ intersects with $D$. Then, $\vert {\cal S}\vert\leq 5$.  
\end{lemma}

\begin{lemma} \label{lem-degen}
Every $K_4$-free disk graph is $11$-degenerate.
\end{lemma}
\begin{proof}
To prove the lemma, it is enough to prove that for any $K_4$-free disk graph, there is a vertex of degree at most $11$. Let $G$ be a $K_4$-free disk graph with a realization $\mathcal{D}$, and let $D_v \in \mathcal{D}$ be the disk representing the vertex $v \in V(G)$. Let $u$ be a vertex in $G$ such that disk $D_u$ has the smallest radius among the disks in $\mathcal{D}$. We will prove that $d_G(u) \leq 11$.  Let $r$ be the radius of $D_u$. Notice that for each $v\in N_G(u)$, the radius of $D_v$ is at least $r$. Now we construct a graph $H$ with vertex set $N_G[u]$ such that $H$ is a unit disk graph, $H$ is a subgraph of $G$ (and hence $K_4$-free), and $d_G(u)=d_{H}(u)$. The construction of $H$ is as follows. For each $v\in N_G(u)$, construct a disk $D'_{v}$ of radius $r$ which is fully contained in the disk $D_v$ and intersects $D_u$. The graph $H$ is the geometric intersection graph of ${\cal D'}=\{D_u\}\cup \{D'_v~\colon v\in N_G(u)\}$. It is easy to see that $H$ is a unit disk graph, $H$ is a subgraph of $G$ and $d_G(u)=d_{H}(u)$. For each $v\in V(H)\setminus \{u\}$, let $L_{v}$ be the line segment between the centers of $D_u$ and $D'_v$.  
Let $\{v_1,\ldots v_t\}=N_G(u)$ such that the line segments $L_{v_1},L_{v_2},\ldots,L_{v_t}$ are in the clockwise order. We claim that $t\leq 11$. For the sake of contradiction assume that $t\geq 12$. Suppose  there exists two distinct $i,j\in \{1,3,5,7,9,11\}$ such that  $D'_{v_{i}}$ intersects with $D'_{v_{j}}$. This implies that 
 $D'_{v_{i+1}}$ or $D'_{v_{j+1}}$ intersects both $D'_{v_{i}}$ and $D'_{v_{j}}$. Let $w\in \{v_{i+1},v_{j+1}\}$ be the vertex such that $D'_{w}$ intersects both $D'_{v_{i}}$ and $D'_{v_{j}}$. Then, $H[\{u,w,v_i,v_j\}]$ is a complete graph on $4$ vertices, which is a contradiction because $H$ is $K_4$-free. Then, the disks $D'_{v_1},D'_{v_3},D'_{v_5},D'_{v_7},D'_{v_9}, D'_{v_{11}}$ are pairwise non-overlapping, which is a contradiction to Lemma~\ref{lem:kissingnumber}.  
Thus, we proved that $t\leq 11$ and hence $d_H(v)=d_G(v)\leq 11$. 
That is, the degeneracy of $G$ is at most $11$. 
\end{proof}

\subsection{The main algorithm}

Our main algorithm for \textsc{Bipartization} on $K_4$-free disk graphs is presented in Algorithm~\ref{alg-oct}.
At the beginning, it takes an arbitrary maximal triangle packing $\mathcal{T}$ of $G$ (line~1) and defines $\mathcal{O}$ as the triangles in $G$ that have at least one vertex outside $V(\mathcal{T})$.
In a high-level, our algorithm computes three different solutions $S_1,S_2,S_3$ and returns the best one.

The first solution $S_1$ is computed in exactly the same way as the 3-approximation algorithm described in the introduction.
Specifically, we include in $S_1$ all vertices in the triangle packing $\mathcal{T}$, and an OCT $X_1$ of $G - V(\mathcal{T})$ computed by a sub-routine \textsc{PlanarBip} (line~3), which is an algorithm for \textsc{Bipartization} on planar graphs.

The second solution $S_2$ is constructed in a more involved way.
First, in each triangle $T \in \mathcal{T}$, we randomly sample one vertex $v_T \in T$ (line~6); here the function $\mathsf{random}(T)$ returns each vertex in $T$ with probability $\frac{1}{3}$ and different calls of $\mathsf{random}$ are independent.
Let $R$ be the vertices in $V(\mathcal{T})$ \textit{not} sampled (line~7).
Then we include in $S_2$ all vertices in $R$, all vertices in an arbitrary maximal triangle packing $\mathcal{T'}$ of $G - R$ (line~8), and an OCT $X_2$ of $G - (R \cup V(\mathcal{T}'))$ computed by \textsc{PlanarBip}.

If $\mathcal{O}$ is nonempty, we need to construct our third solution $S_3$.
We take a maximal packing $\mathcal{T}''$ of the triangles in $\mathcal{O}$ (line~12).
Then we include in $S_3$ all vertices in $V(\mathcal{T}'') \cap V(\mathcal{T})$ and an OCT $X_3$ of $G-(V(\mathcal{T}'') \cap V(\mathcal{T}))$ recursively computed by our algorithm (line~13).

Finally, we return the best one among $S_1,S_2,S_3$ (line~15); here $\min\{S_1,S_2,S_3\}$ returns the set of smallest size among $S_1,S_2,S_3$.
If $S_3$ is not computed, we simply return $\min\{S_1,S_2\}$.
It is obvious that each of $S_1,S_2,S_3$ is an OCT of $G$ and thus the algorithm always returns a correct solution.
The quality of the solution obtained will be analyzed in the next section.
Also, one can easily see that Algorithm~\ref{alg-oct} runs in polynomial time.
Indeed, except the recursive call of $\textsc{OCT}$ in line~13, all the other steps can be done in polynomial time.
Line~13 will only be executed when $\mathcal{O} \neq \emptyset$.
In this case, $\mathcal{T}'' \neq \emptyset$ and $V(\mathcal{T}'') \neq \emptyset$.
Thus, the graph $G-(V(\mathcal{T}'') \cap V(\mathcal{T}))$ has at most $n-1$ vertices where $n = V(G)$, and the running time of Algorithm~\ref{alg-oct} satisfies the recurrence $T(n) \leq T(n-1) + n^{O(1)}$ which solves to $T(n) = n^{O(1)}$.

\begin{algorithm}[h]
    \caption{$\textsc{Bipartization}(G)$ \Comment{$G$ is a $K_4$-free disk graph}}
	\begin{algorithmic}[1]
	    \State $\mathcal{T} \leftarrow$ a maximal packing of $\Delta(G)$
	    \State $\mathcal{O} \leftarrow \{T \in \Delta(G): T \nsubseteq V(\mathcal{T})\}$
	    \medskip
	    \State $X_1 \leftarrow \textsc{PlanarBip}(G-V(\mathcal{T}))$ \Comment{construct the first solution $S_1$}
	    \State $S_1 \leftarrow V(\mathcal{T}) \cup X_1$
	    \medskip
	    \For{every $T \in \mathcal{T}$} \Comment{construct the second solution $S_2$}
	        \State $v_T \leftarrow \mathsf{random}(T)$
	    \EndFor
	    \State $R \leftarrow \bigcup_{T \in \mathcal{T}} (T \backslash \{v_T\})$ 
	    \State $\mathcal{T}' \leftarrow$ a maximal packing of $\Delta(G-R)$
	    \State $X_2 \leftarrow \textsc{PlanarBip}(G-(R \cup V(\mathcal{T}')))$
	    \State $S_2 \leftarrow R \cup V(\mathcal{T}') \cup X_2$
	    \medskip
	    \If{$\mathcal{O} \neq \emptyset$} \Comment{construct the third solution $S_3$}
	        \State $\mathcal{T}'' \leftarrow$ a maximal packing of $\mathcal{O}$
	        \State $X_3 \leftarrow \textsc{Bipartization}(G-(V(\mathcal{T}'') \cap V(\mathcal{T})))$
	        \State $S_3 \leftarrow (V(\mathcal{T}'') \cap V(\mathcal{T})) \cup X_3$
	    \EndIf
	    \medskip
	    \State \textbf{return} $\min\{S_1,S_2,S_3\}$ \Comment{if $S_3$ is undefined, simply return $\min\{S_1,S_2\}$}
	\end{algorithmic}
	\label{alg-oct}
\end{algorithm}

\subsection{Analysis}
In this section, we analyze the (expected) approximation ratio of Algorithm~\ref{alg-oct}.
We denote by $\rho$ this ratio and aim to establish an upper bound for $\rho$.
Consider a given disk graph $G$ which is $K_4$-free.
Let $\mathsf{opt}$ be the minimum size of an odd cycle transversal of $G$, and $\mathcal{T}, \mathcal{O}$ be the two sets of triangles as in Algorithm~\ref{alg-oct}.
The output of Algorithm~\ref{alg-oct} is the best one among three OCT solutions $S_1,S_2,S_3$.
Therefore, to analyze the approximation ratio of our algorithm, we have to consider the approximation ratios of $S_1,S_2,S_3$.
It turns out that each solution $S_i$ individually may be of size $3 \mathsf{opt}$ or even larger in worst case.
However, as we will see, the best one among them always admits an approximation ratio strictly smaller than 3.

In order to analyze the three solutions, we define two important parameters: $a = |\mathcal{T}|/\mathsf{opt}$ and $b = \mathsf{tri}(\mathcal{O})/\mathsf{opt}$ (recall that $\mathsf{tri}(\mathcal{O})$ is the size of a \textit{maximum} packing of $\mathcal{O}$).
Note that $a,b \in [0,1]$ because both $|\mathcal{T}|$ and $\mathsf{tri}(\mathcal{O})$ are at most the size of a maximum triangle packing in $G$, which is smaller than or equal to $\mathsf{opt}$.
The analysis of $S_1,S_2,S_3$ will be done in terms of $a$ and $b$, that is, we shall represent the approximation ratios of $S_1,S_2,S_3$ as functions of $a$ and $b$.
Roughly speaking, we shall show that $S_1$ is good when $a$ is small, $S_2$ is good when $b$ is small, and $S_3$ is good when both $a$ and $b$ are large.
For convenience, we use the notation $\rho_0$ to denote the approximation ratio of the \textsc{PlanarBip} sub-routine used in Algorithm~\ref{alg-oct}.

\subsubsection{The quality of $S_1$}


The solution $S_1$ is computed using exactly the 3-approximation algorithm described in the introduction.
A more careful analysis shows that its approximation ratio is related to the parameter $a$: the smaller $a$ is, the better $S_1$ is.

\begin{lemma} \label{lem-S1}
$|S_1| \leq (3a+\rho_0(1-a)) \cdot \mathsf{opt}$.
\end{lemma}

\begin{proof}
Since $a=|\mathcal{T}|/\mathsf{opt}$, and $\mathcal{T}$ is a triangle packing, we get that  $|V(\mathcal{T})|=3|\mathcal{T}|=3a\cdot \mathsf{opt}$. Since $\mathcal{T}$ is a triangle packing, any odd cycle transversal contains at least $|\mathcal{T}|$ vertices from 
$V(\mathcal{T})$, the size of a minimum odd cycle transversal in $G-V(\mathcal{T})$ is at most $\mathsf{opt}-|\mathcal{T}|=(1-a)\mathsf{opt}$. Therefore, 
$|S_1|=|V(\mathcal{T})|+|X_1|\leq (3a+\rho_0(1-a)) \cdot \mathsf{opt}$. 
\end{proof}

\subsubsection{The quality of $S_2$}

Figuring out the quality of $S_2$ is the most involved part in our analysis.
Basically, what we shall show is that whenever the parameter $b$ is sufficiently small, $S_2$ always gives us a better-than-3 approximation no matter what the value of $a$ is.
The analysis in this section shall explicitly use the fact that $G$ is $K_4$-free.
Let $v_T$, $R$, and $\mathcal{T}'$ be as in Algorithm~\ref{alg-oct}.
We first need the following simple observation, which will allow us to bound the expected size of $S_2$ using the expected size of $\mathcal{T}'$.

\begin{observation}
\label{lem:largehitopt}
$\mathbb{E}[|R \cap S_\textnormal{opt}|] \geq \frac{1}{3} |R|$ and $|V(\mathcal{T}') \cap S_\textnormal{opt}| \geq |\mathcal{T}'|$.
\end{observation}
\begin{proof}
Since $S_\textnormal{opt}$ is an OCT of $G$, it contains at least one vertex in every triangle $T \in \mathcal{T}$.
Thus, $\mathbb{E}[|(T \backslash \{v_T\}) \cap S_\textnormal{opt}|] \geq \frac{2}{3}$.
By the linearity of expectation, we have
\begin{equation*}
    \mathbb{E}[|R \cap S_\textnormal{opt}|] = \sum_{T \in \mathcal{T}} \mathbb{E}[|(T \backslash \{v_T\}) \cap S_\textnormal{opt}|] \geq \frac{2}{3} |\mathcal{T}| = \frac{1}{3}|R|.
\end{equation*}
The fact $|V(\mathcal{T}') \cap S_\textnormal{opt}| \geq |\mathcal{T}'|$ follows from the fact that $\mathcal{T}'$ is a triangle packing: $S_\textnormal{opt}$ contains at least one vertex in every triangle $T \in \mathcal{T}'$.
\end{proof}

The above observation implies that $\mathbb{E}[|(R \cup V(\mathcal{T}')) \cap S_\textnormal{opt})|] \geq \frac{1}{3} |R| + \mathbb{E}[|\mathcal{T}'|]$.
Therefore, we have $\mathbb{E}[|S_\text{opt} \backslash (R \cup V(\mathcal{T}'))|] \leq \mathsf{opt} - \frac{1}{3} |R| - \mathbb{E}[|\mathcal{T}'|]$ and hence $\mathbb{E}[|X_2|] \leq \rho_0 \cdot (\mathsf{opt} - \frac{1}{3} |R| - \mathbb{E}[|\mathcal{T}'|])$.
It follows that
\begin{equation} \label{eq-S2T'}
    \begin{aligned}
        \mathbb{E}[|S_2|] &= |R| + \mathbb{E}[|V(\mathcal{T}')|] + \mathbb{E}[|X_2|] \\
        &\leq |R|+ 3 \cdot \mathbb{E}[|\mathcal{T}'|] + \rho_0 \cdot \left(\mathsf{opt} - \frac{1}{3} |R| - \mathbb{E}[|\mathcal{T}'|]\right) \\
        &= \left(1 - \frac{\rho_0}{3}\right) \cdot |R| + \rho_0 \cdot \mathsf{opt} + (3 - \rho_0) \cdot \mathbb{E}[|\mathcal{T}'|] \\
    \end{aligned}
\end{equation}

We say a vertex $v \in V(\mathcal{T})$ is \textit{dead} if $v \notin R$ and $v$ is not contained in any triangle in $G[V(\mathcal{T}) \backslash R]$.
Let $D$ denote the set of all dead vertices, which is a random subset of $V(\mathcal{T})$ as it depends on the random vertices $v_T$.
For each $v \in V(\mathcal{T})$, let $\mathsf{deg}(v)$ denote the degree of $v$ in $G[V(\mathcal{T})]$.
Recall that $a = |\mathcal{T}|/\mathsf{opt}$ and $b = \mathsf{tri}(\mathcal{O})/\mathsf{opt}$.
It is easy to see the following relation between $|\mathcal{T}'|$ and $|D|$.

\begin{observation} \label{obs-T'dead}
$|\mathcal{T}'| \leq (a + b) \cdot \mathsf{opt} - \frac{|R|}{3} - \frac{|D|}{3}$.
\end{observation}
\begin{proof}
Since $\mathcal{T}'$ is a triangle packing, we have $|\mathcal{T}' \cap \mathcal{O}| \leq \mathsf{tri}(\mathcal{O}) = b \cdot \mathsf{opt}$.
On the other hand, all elements in $\mathcal{T}' \backslash \mathcal{O}$ are triangles in $G[V(\mathcal{T}) \backslash R]$.
However, dead vertices cannot be the vertex of any triangle in $G[V(\mathcal{T}) \backslash R]$.
Thus, the vertices of the triangles in $\mathcal{T}' \backslash \mathcal{O}$ must lie in $V(\mathcal{T}) \backslash (R \cup D)$.
We have $|V(\mathcal{T}) \backslash (R \cup D)| = 3 \cdot |\mathcal{T}| - |R| - |D| = 3a \cdot \mathsf{opt} - |R| - |D|$, which implies $|\mathcal{T}' \backslash \mathcal{O}| \leq a \cdot \mathsf{opt} - \frac{|R|}{3} - \frac{|D|}{3}$.
Because $|\mathcal{T}'| = |\mathcal{T}' \cap \mathcal{O}|+|\mathcal{T}' \backslash \mathcal{O}|$, we have the inequality in the observation.
\end{proof}

Combining the above observation with Equation~\ref{eq-S2T'}, we have
\begin{equation} \label{eq-S2D}
    \begin{aligned}
         \mathbb{E}[|S_2|] &\leq \left(1 - \frac{\rho_0}{3}\right) \cdot |R| + \rho_0 \cdot \mathsf{opt} + (3 - \rho_0) \cdot \mathbb{E}[|\mathcal{T}'|] \\
         &= \left(1 - \frac{\rho_0}{3}\right) \cdot |R| + \rho_0 \cdot \mathsf{opt} + (3 - \rho_0) \cdot \left( (a + b) \cdot \mathsf{opt} - \frac{|R|}{3} - \frac{\mathbb{E}[D]}{3} \right) \\
         &= (\rho_0 + (3 - \rho_0) (a+b)) \cdot \mathsf{opt} - (3 - \rho_0) \cdot \frac{\mathbb{E}[|D|]}{3}.
    \end{aligned}
\end{equation}

To show that $S_2$ has an approximation ratio below 3 when $b$ is small, the crucial observation is that we have a large number of dead vertices in expectation.

\begin{observation} \label{obs-dead}
$\Pr[v \textnormal{ is dead}] \geq (\frac{1}{3})^{\frac{3}{8}\mathsf{deg}(v)+\frac{1}{4}}$ for all $v \in V(\mathcal{T})$.
Thus, we have $\mathbb{E}[|D|] \geq (\frac{1}{3})^{\frac{3}{8} d+\frac{1}{4}} (3a \cdot \mathsf{opt})$, where $d$ is the average degree of $G[V(\mathcal{T})]$.
\end{observation}
\begin{proof}
Let $v \in V(\mathcal{T})$ and $T_0 \in \mathcal{T}$ be the triangle containing $v$.
Denote by $N(v)$ the set of neighbors of $v$ in $V(\mathcal{T})$ (excluding $v$ itself).
We then have $|N(v)| = \mathsf{deg}(v)$ and $|N(v) \backslash T_0| = \mathsf{deg}(v) - 2$.
Observe that the graph $G[N(v)]$ is triangle-free.
Indeed, if $G[N(v)]$ contains a triangle $T$, then $T \cup \{v\}$ forms a clique in $G$ of size 4, which contradicts with the fact that $G$ is $K_4$-free.
As $G[N(v)]$ is triangle-free, it is planar.
In particular, $G[N(v) \backslash T_0]$ is planar.
It was known that every $n$-vertex planar graph has a vertex cover of size at most $\frac{3}{4}n$ (indeed, a planar graph is 4-colorable, so it has an independent set of size at least $\frac{n}{4}$ and thus a vertex cover of size at most $\frac{3}{4}n$).
Thus, $G[N(v) \backslash T_0]$ has a vertex cover $C \subseteq N(v) \backslash T_0$ with $|C| \leq \frac{3}{4}|N(v) \backslash T_0| = \frac{3}{4} (\mathsf{deg}(v)-2)$.

Next, we notice that if $v \notin R$ and $C \subseteq R$, then $v$ is a dead vertex.
Indeed, if $v$ is contained in a triangle $\{v,u,w\}$ in $G[V(\mathcal{T}) \backslash R]$, then at least one of $u$ and $w$ must be in $C$, since $C$ is a vertex cover of $N(v) \backslash T_0$ (note that $u,w \notin T_0$ for otherwise $v \in R$).
Let $\mathcal{T}_1 \subseteq \mathcal{T}$ (resp., $\mathcal{T}_2 \subseteq \mathcal{T}$) consist of the triangles that contain one vertex (resp., two vertices) in $C$.
Note that no triangle in $\mathcal{T}$ can contain three vertices in $C$ because $G[C]$ is triangle-free.
Therefore, $C \subseteq R$ if and only if $v_T \notin C$ for all $T \in \mathcal{T}_1 \cup \mathcal{T}_2$.
The events $v_T \notin C$ for all $T \in \mathcal{T}_1 \cup \mathcal{T}_2$ are independent, and happen with probability $\frac{2}{3}$ if $T \in \mathcal{T}_1$ and with probability $\frac{1}{3}$ if $T \in \mathcal{T}_2$.
It follows that $\Pr[C \subseteq R] = (\frac{2}{3})^{|\mathcal{T}_1|} \cdot (\frac{1}{3})^{|\mathcal{T}_2|}$.
We have the inequality $|\mathcal{T}_1| + 2|\mathcal{T}_2| = |C| \leq \frac{3}{4} (\mathsf{deg}(v)-2)$.
Subject to $|\mathcal{T}_1| \geq 0$, $|\mathcal{T}_2| \geq 0$, and $|\mathcal{T}_1| + 2|\mathcal{T}_2| \leq \frac{3}{4} (\mathsf{deg}(v)-2)$, the quantity $(\frac{2}{3})^{|\mathcal{T}_1|} \cdot (\frac{1}{3})^{|\mathcal{T}_2|}$ is minimized when $|\mathcal{T}_1| = 0$ and $|\mathcal{T}_2| = \frac{3}{8} (\mathsf{deg}(v)-2)$, and is equal to $(\frac{1}{3})^{\frac{3}{8} (\mathsf{deg}(v)-2)}$.
Thus, we have $\Pr[C \subseteq R] \geq (\frac{1}{3})^{\frac{3}{8} (\mathsf{deg}(v)-2)}$.
Furthermore, the events $v \notin R$ and $C \subseteq R$ are independent because whether $v \notin R$ happens only depends on the choice of $v_{T_0} \in T_0$.
We have $\Pr[v \notin R] = \frac{1}{3}$ and $\Pr[C \subseteq R] \geq (\frac{1}{3})^{\frac{3}{8} (\mathsf{deg}(v)-2)}$.
Since $v$ is a dead vertex if both $v \notin R$ and $C \subseteq R$ happen, we finally have $\Pr[v \textnormal{ is dead}] \geq (\frac{1}{3})^{\frac{3}{8}\mathsf{deg}(v)+\frac{1}{4}}$.
By the linearity of expectation and the fact $|V(\mathcal{T})| = 3|\mathcal{T}| = 3a \cdot \mathsf{opt}$, we then have
\begin{equation*}
    \mathbb{E}[|D|] = \sum_{v \in V(\mathcal{T})} \Pr[v \textnormal{ is dead}] \geq \sum_{v \in V(\mathcal{T})} \left(\frac{1}{3}\right)^{\frac{3}{8}\mathsf{deg}(v)+\frac{1}{4}} \geq \left(\frac{1}{3}\right)^{\frac{3}{8}d+\frac{1}{4}} (3a \cdot \mathsf{opt}).
\end{equation*}
where $d = \sum_{v \in V(\mathcal{T})} \mathsf{deg}(v) / |V(\mathcal{T})|$.
\end{proof}

At the end, we can establish the bound for $\mathbb{E}[|S_2|]$.

\begin{lemma} \label{lem-S2}
$\mathbb{E}[|S_2|] \leq (\rho_0 + (1-(\frac{1}{3})^{\frac{3}{8}d+\frac{1}{4}})(3-\rho_0) a + (3- \rho_0) b) \cdot \mathsf{opt}$, where $d$ denotes the average degree of $G[V(\mathcal{T})]$.
\end{lemma}
\begin{proof}
Combining Observation~\ref{obs-dead} with Equation~\ref{eq-S2D} completes the proof.
\end{proof}

According to Lemma~\ref{lem-degen}, we have $d \leq 22$, and thus the above lemma gives us a good bound for $\mathbb{E}[|S_2|]$: as long as $b$ is sufficiently small, no matter what $a$ is, we can have that $\mathbb{E}[|S_2|] \leq (3-\alpha) \cdot \mathsf{opt}$ for some constant $\alpha > 0$.

\subsubsection{The quality of $S_3$}

Finally, we analyze the quality of $S_3$.
Given $S_1$ is good when $a$ is small and $S_2$ is good when $b$ is small, we clearly want $S_3$ to be good when both $a$ and $b$ are large.

\begin{lemma} \label{lem-S3}
If $\rho \geq 2$, then $\mathbb{E}[|S_3|] \leq (\frac{2b}{3}+\rho(2-a-\frac{b}{3})) \cdot \mathsf{opt}$.
\end{lemma}
\begin{proof}
Let $r=|\mathcal{T}''|/\mathsf{opt}$.
Since $\mathcal{T}''$ is a \textit{maximal} packing in $\mathcal{O}$, $V(\mathcal{T}'')$ is a hitting set of $\mathcal{O}$, which implies $3 |\mathcal{T}''| = |V(\mathcal{T}'')| \geq \mathsf{tri}(\mathcal{O})$ and thus $r\geq \frac{\mathsf{tri}(\mathcal{O})}{3 \cdot  \mathsf{opt}}=\frac{b}{3}$.

\begin{claim}
\label{clim:X3}
$\mathbb{E}[|X_3|] \leq \rho (2-a-r)\mathsf{opt}$. 
\end{claim}

\begin{proof}
Let $S^*$ be an optimal OCT of $G$. 
So $|S^*|=\mathsf{opt}$. 
We call a triangle $T \in \mathcal{T}''$ {\em bad} if $S^*$ contains a vertex from $V(T)\setminus V(\mathcal{T})$.
Since $\mathcal{T}$ is a triangle packing, we also know that $|S^*\cap V(\mathcal{T})|\geq |\mathcal{T}|=a \cdot \mathsf{opt}$.
Thus, the number of bad triangles in $\mathcal{T}''$ is at most $(1-a)\mathsf{opt}$, because $S^*$ contains at most $(1-a)\mathsf{opt}$ vertices outside of $V(\mathcal{T})$, and any such vertex can be part of at most one triangle in $\mathcal{T}''$. 
That is, the number of \textit{good} triangles in $\mathcal{T}''$ is at least $(r-(1-a))\mathsf{opt}$.
For each good triangle $T \in \mathcal{T}''$, $S^*$ contains a vertex from $V(T)\cap V(\mathcal{T})$.
This implies that $|S^*\cap (V(\mathcal{T}'')\cap V(\mathcal{T}))| \geq (r-(1-a))\mathsf{opt}$, and hence $|S^* \setminus (V(\mathcal{T}'')\cap V(\mathcal{T}))| \leq \mathsf{opt}- (r-(1-a))\mathsf{opt}= (2-a-r)\mathsf{opt}$.
Also, notice that $S^*\setminus (V(\mathcal{T}'')\cap V(\mathcal{T}))$ is an OCT of $G-(V(\mathcal{T}'')\cap V(\mathcal{T}))$.
Thus, the size of an optimal OCT of $G-(V(\mathcal{T}'')\cap V(\mathcal{T}))$ is at most $(2-a-r)\mathsf{opt}$, which implies $\mathbb{E}[|X_3|] \leq \rho (2-a-r)\mathsf{opt}$.
\end{proof}
Since each triangle $T \in \mathcal{T}''$ is also in $\mathcal{O}$, we have $|V(T)\cap V(\mathcal{T})|\leq 2$.
This implies that $|V(\mathcal{T}'') \cap V(\mathcal{T})| \leq 2r \cdot \mathsf{opt}$.
Now we are ready to deduce
\begin{eqnarray*}
\mathbb{E}[|S_3|] &=& \mathbb{E}[|V(\mathcal{T}'') \cap V(\mathcal{T})|+|X_3|] \\
&=& |V(\mathcal{T}'') \cap V(\mathcal{T})| + \mathbb{E}[|X_3|] \\
&\leq & (2r \cdot \mathsf{opt}) + \rho (2-a-r)\mathsf{opt} \qquad\qquad (\mbox{By Claim~\ref{clim:X3} and }|V(\mathcal{T}'') \cap V(\mathcal{T})|\leq 2r \cdot \mathsf{opt})\\
&=& (2-\rho)r\cdot \mathsf{opt}+ \rho (2-a) \mathsf{opt} \\
&\leq& (2-\rho) \frac{b}{3} \mathsf{opt}+  \rho (2-a) \mathsf{opt} \qquad \qquad (\mbox{Because } r\geq \frac{b}{3}, \mbox{ and } 2-\rho<0)\\
&\leq & \left(\frac{2b}{3}+\rho\left(2-a-\frac{b}{3}\right)\right)\cdot \mathsf{opt}.
\end{eqnarray*}
This completes the proof of the lemma.
\end{proof}

\subsubsection{Putting everything together} \label{sec-analysis}
Given the analyses for $S_1$, $S_2$, and $S_2$ in the previous sections, we are ready to bound the (expected) approximation ratio $\rho$ of the entire algorithm.
Let $\rho_1 = 3a+\rho_0(1-a)$, $\rho_2 = \rho_0 + (1-(\frac{1}{3})^{\frac{3}{8}d+\frac{1}{4}})(3-\rho_0) a + (3- \rho_0) b$, $\rho_3 = \frac{2b}{3}+\rho(2-a-\frac{b}{3})$ be the approximation ratios of $S_1,S_2,S_3$ given in Lemmas~\ref{lem-S1}, \ref{lem-S2}, \ref{lem-S3}, respectively\footnote{In Lemma~\ref{lem-S3}, the bound for $\mathbb{E}[|S_3|]$ has a condition $\rho \geq 2$. But we can assume this is always the case, for otherwise our algorithm is already a 2-approximation algorithm.}.
As the output is the best one among $S_1,S_2,S_3$, we have
\begin{equation*}
    \rho \leq \frac{\mathbb{E}[\min\{|S_1|,|S_2|,|S_3|\}]}{\mathsf{opt}} \leq \frac{\min\{\mathbb{E}[|S_1|],\mathbb{E}[|S_2|],\mathbb{E}[|S_3|]\}}{\mathsf{opt}} \leq \min\{\rho_1,\rho_2,\rho_3\}.
\end{equation*}
Note that $\rho_1,\rho_2,\rho_3$ can be viewed as linear functions of $a$ and $b$, when the other numbers $d$, $\rho_0$, $\rho$ are all fixed.
So we first figure out the values of $a$ and $b$ that maximizes $\min\{\rho_1,\rho_2,\rho_3\}$.
With calculation, we have
\begin{equation} \label{eq-min}
    \min\{\rho_1,\rho_2,\rho_3\} \leq \frac{(3-\rho_0)(2\rho - \rho_0)}{(3-\rho_0+\rho) + (\rho-2) \cdot 3^{-(\frac{3}{8}d + \frac{5}{4})}} + \rho_0,
\end{equation}
and the upper bound is achieved when
\begin{equation*}
    a = \frac{(2\rho - \rho_0)}{(3-\rho_0+\rho) + (\rho-2) \cdot 3^{-(\frac{3}{8}d + \frac{5}{4})}} \text{ and } b = \frac{(2\rho - \rho_0) \cdot 3^{\frac{3}{8}d + \frac{1}{4}}}{(3-\rho_0+\rho) + (\rho-2) \cdot 3^{-(\frac{3}{8}d + \frac{5}{4})}},
\end{equation*}
in which case $\rho_1 = \rho_2 = \rho_3$.
Now combine Equation~\ref{eq-min} with the inequality $\rho \leq \min\{\rho_1,\rho_2,\rho_3\}$ and re-arrange the terms in the inequality, we deduce
\begin{equation*}
    (3^{-(\frac{3}{8}d + \frac{5}{4})}+1) \cdot \rho^2 - ((2+\rho_0) \cdot 3^{-(\frac{3}{8}d + \frac{5}{4})} +3) \cdot \rho + (2 \rho_0 \cdot 3^{-(\frac{3}{8}d + \frac{5}{4})}) \leq 0.
\end{equation*}
The left-hand side of the above inequality is a quadratic function of $\rho$ in which the coefficient of the quadratic term is positive.
Therefore, in order to make the quadratic function non-positive, $\rho$ must be smaller than its larger root, i.e.,
\begin{equation*}
    \rho \leq \frac{((2+\rho_0) \cdot 3^{-(\frac{3}{8}d + \frac{5}{4})} +3) + \sqrt{((2+\rho_0) \cdot 3^{-(\frac{3}{8}d + \frac{5}{4})} +3)^2 - (3^{-(\frac{3}{8}d + \frac{5}{4})}+1)(8 \rho_0 \cdot 3^{-(\frac{3}{8}d + \frac{5}{4})})}}{2 \cdot (3^{-(\frac{3}{8}d + \frac{5}{4})}+1)}.
\end{equation*}
By Lemma~\ref{lem-degen}, $G[V(\mathcal{T})]$ is $11$-degenerate and thus $d \leq 22$.
Furthermore, using the $\frac{9}{4}$-approximation algorithm \cite{GoemansW98} for planar bipartization, we can set $\rho_0 = \frac{9}{4}$.
Plugging in these values to the above inequality, we have $\rho \leq 2.99993033741$.

Our entire algorithm first applies Lemma~\ref{lem:k4reduction} with a maximal packing $\mathcal{C}$ of $K_4(G)$ to reduce the problem to $G - V(\mathcal{C})$, which is a $K_4$-free disk graph, and then applies Algorithm~\ref{alg-oct} on $G - V(\mathcal{C})$.
By Lemma~\ref{lem:k4reduction} and the above analysis, this algorithm solves \textsc{Bipartization} on disk graphs with an \textit{expected} approximation ratio at most $2.99993033741$.
By repeating the algorithm polynomial number of times, we can also obtain a randomized algorithm that achieves the same approximation ratio with high probability.

\begin{theorem}
There exists a polynomial-time randomized algorithm for \textsc{Bipartization} on the class of disk graphs that gives a $(3-\alpha)$-approximation solution with high probability, for some $\alpha > 10^{-5}$.
\end{theorem}

With some efforts, we can derandomize our algorithm to obtain a deterministic $(3-\alpha)$-approximation algorithm for \textsc{Bipartization} on disk graphs.
In fact, in Algorithm~\ref{alg-oct}, only the construction of the set $R$ is randomized.
We show in Appendix~\ref{sec-derandom} how to construct $R$ deterministically while still guaranteeing the nice properties of $R$ (the key point is to guarantee that there are many dead vertices).
The approximation ratio of our deterministic algorithm is slightly worse than the randomized one (while it is still smaller than 3).

\main*

\section{Generalizations} \label{sec-generalization}
In this section, we discuss some generalizations of our result.
First, we observe that our algorithm directly generalizes to \textit{pseudo-disk graphs}.
A set of geometric objects in the plane are called \textit{pseudo-disks} if each of them is homeomorphic to a disk and the boundaries of any two objects intersect at most twice.
A graph is a \textit{pseudo-disk graph} if it can be represented as the intersection graph of a set of pseudo-disks.

In our \textsc{Bipartization} algorithm, we only exploit two properties of disk graphs: (i) triangle-free disk graphs are planar, and (ii) $K_4$-free disk graphs are $c$-degenerate for some constant $c$.
In fact, pseudo-disk graphs also satisfy these two properties.

\begin{fact}[\cite{kratochvil1996intersection}]
Triangle-free pseudo-disk graphs are planar.
\end{fact}

\begin{fact}
$K_4$-free pseudo-disk graphs are $c$-degenerate for some constant $c$.
\end{fact}
\begin{proof}
We show that any $K_r$-free pseudo-disk graph of $n$ vertices only has $O(rn)$ edges, which implies the fact.
Let $G$ be a $K_r$-free pseudo-disk graph realized by a set $\mathcal{S}$ of pseudo-disks.
We say an edge $(S,S')$ of $G$ is an \textit{inclusion edge} if $S \subseteq S'$ or $S' \subseteq S$.
We first observe that $G$ has $O(rn)$ inclusion edges.
Indeed, a pseudo-disk $S \in \mathcal{S}$ cannot be contained in $r-1$ (or more) other pseudo-disks in $\mathcal{S}$, for otherwise there is a copy of $K_r$ in $G$.
Thus, if we charge every inclusion edge $(S,S')$ to the smaller one of $S$ and $S'$, every pseudo-disk is charged at most $r-2$ times.
This implies that $G$ has $O(rn)$ inclusion edges.
Now we bound the number of other edges in $G$.
Note that if $(S,S')$ is a non-inclusion edge in $G$, then the boundaries of $S$ and $S'$ intersect.
So it suffices to bound the total number of intersection points of the boundaries of the pseudo-disks in $\mathcal{S}$.
The depth of an intersection point $x$ is the number of pseudo-disks in $\mathcal{S}$ containing $x$.
It is well-known that in a set of $n$ pseudo-disks, the number of boundary intersection points of depth at most $d$ is bounded by $O(dn)$~\cite{kedem1986union}.
Since $G$ is $K_r$-free, every intersection point is of depth at most $r$.
Thus, the total number of boundary intersection points is $O(rn)$, implying that $G$ has $O(rn)$ edges.
\end{proof}

Therefore, our algorithm directly generalizes to pseudo-disk graphs.

\mainpseudo*


Next, we observe that our techniques apply to not only the specific problem of \textsc{Bipartization}.
In fact, it works for a wide class of vertex-deletion problems on (pseudo-)disk graphs.
Recall that in a vertex-deletion problem, the goal is to delete a minimum set $S$ of vertices from a graph $G$ such that $G-S$ satisfies some desired property $\textbf{P}$.
In \textsc{Bipartization}, the property $\textbf{P}$ is ``being bipartite''.
Our technique applies to any vertex-deletion problem that is \textbf{(i)} hereditary, i.e., if a graph satisfies $\textbf{P}$ then all its induced subgraphs also satisfy $\textbf{P}$, and \textbf{(ii)} triangle-conflicting, i.e., a graph satisfies $\textbf{P}$ only if it is triangle-free.
\begin{theorem} \label{thm-general}
If a hereditary and triangle-conflicting vertex-deletion problem admits a $(3-\delta)$-approximation algorithm on the class of planar graphs for some $\delta > 0$, then it admits a $(3-\alpha)$-approximation algorithm on the class of (pseudo-)disk graphs for some $\alpha > 0$.
\end{theorem}
\begin{proof}
We just replace the \textsc{PlanarBip} sub-routine in Algorithm~\ref{alg-oct} with the $(3-\delta)$-approximation algorithm for the problem on planar graphs.
As one can easily verify, our analysis only depends on the fact that the vertex-deletion problem is hereditary and triangle-conflicting.
Thus, the same analysis shows that this is a $(3-\alpha)$-approximation algorithm for the problem on (pseudo-)disk graphs.
\end{proof}

Well-studied instances of hereditary and triangle-conflicting vertex-deletion problems (other than \textsc{Bipartization}) include \textsc{Vertex Cover}, \textsc{Feedback Vertex Set}, \textsc{Triangle Hitting}, etc.
Most of these problems already have $(3-\delta)$-approximation algorithms on disk graphs, and some of them do not have known $(3-\delta)$-approximation algorithms on planar graphs.
Thus, at this point, we can only obtain interesting results for \textsc{Bipartization}.
However, we believe that this is not the full power of Theorem~\ref{thm-general}.
To provide some evidences, let us consider a vertex-deletion problem, \textsc{Planarization\&Bipartization}, in which the desired property $\textbf{P}$ is ``being planar and bipartite''.
This problem is clearly hereditary and triangle-conflicting.
On planar graphs, it is equivalent to \textsc{Bipartization} and thus admits a $(3-\delta)$-approximation algorithm.
Therefore, Theorem~\ref{thm-general} gives a $(3-\alpha)$-approximation algorithm for \textsc{Planarization\&Bipartization} on (pseudo-)disk graphs.
Although this problem itself is somehow artificial and not well-studied, it reveals that Theorem~\ref{thm-general} could possibly have further applications in the future.

\bibliographystyle{plain}
\bibliography{my_bib.bib}

\appendix

\section{Derandomization} \label{sec-derandom}
The only randomized step in Algorithm~\ref{alg-oct} is Step~6 and hence the set $R$ is a random variable.
Towards derandomization we construct three sets $R_1,R_2$, and $R_3$ such that at least one of them is ``good''.
Next we explain the construction of these sets $R_1,R_2$, and $R_3$. Let $H'=\{v\in V({\cal T})~:~ d_{G[V({\cal T})]}(v)>100\}$ be the set of ``high'' degree vertices in $G[V({\cal T})]$. Let ${\cal H}$ be the set of triangles $T$ in ${\cal T}$ such that $T\cap H' \neq \emptyset$. Let ${\cal L}={\cal T}\setminus {\cal H}$.   
Let $H=V({\cal H})$ and $L=V({\cal L})=V({\cal T}) \setminus H$.

\begin{itemize}
    \item Let $I$ be a maximal distance-$3$ independent set in $G[L]$ constructed as follows. Initially, set $I:=\emptyset$ and $C:=L$. As long as $C\neq \emptyset$, pick a vertex $v$ from $C$ and add it to $I$, and delete $N_{G[L]}[N_{G[L]}[N_{G[L]}[v]]]$ from $C$. 
    \item Let $I_1,I_2,I_3$ be arbitrary subsets of $I$ such that they are pairwise disjoint and $|I_i|=\lfloor \frac{I}{3} \rfloor$ for all $i\in \{1,2,3\}$.  
\end{itemize}
Before explaining the construction of $R_1,R_2$, and $R_3$, we prove the following lemma. 

\begin{lemma}
\label{lem:atmost2}
Let $i\in \{1,2,3\}$ and $v\in I_i$. Let $T$ be a triangle in ${\cal T}$ such that $N_G(v)\cap T\neq \emptyset$. Then, there is at least one vertex $w\in T$ such that $w\notin N_G(v)$. 
\end{lemma}

\begin{proof}
If $v\in T$, then $v \notin N_G(v)$ and $v$ is the required vertex. Otherwise, we have that $v\notin T$. Then, if all the vertices of $T$ are adjacent to $v$, the there is a $K_4$ with vertex set $T\cup \{v\}$. This is a contradiction to the fact that $G$ is $K_4$-free. This completes the proof of the lemma.  
\end{proof}

Now we construct $R_1,R_2$, and $R_3$ as follows. 

\begin{itemize}
    \item[(a)] Initially, set $R_i:=H\cup N_{G[L]}(I_i)$ for all $i\in \{1,2,3\}$. 
    \item[(b)] Because $I$ is a distance-$3$ independent set, for each triangle $T\in {\cal L}$, there is at most one index $i\in \{1,2,3\}$ such that $N_{G[L]}(I_i)\cap T \neq \emptyset$. 
    \item[(c)] Now, for each triangle $T$ in ${\cal L}$ we do the following. Let $T=\{x_1,x_2,x_3\}$. 
    \begin{itemize}
        \item Suppose $T\cap R_i=\emptyset$ for all $i\in \{1,2,3\}$. Then, $R_1:=R_1\cup \{x_1,x_2\}$, $R_2:=R_2\cup \{x_2,x_3\}$, and $R_3:=R_3\cup \{x_3,x_1\}$. 
        \item Otherwise, there is exactly one index $i\in \{1,2,3\}$ such that $N_{G[L]}(I_i)\cap T \neq \emptyset$. Moreover, since $I_i$ is a distance-$3$ independent set, there is exactly one vertex $v\in I_i$ such that $N_{G[L]}(v)\cap T \neq \emptyset$. Without loss of generality, by Lemma~\ref{lem:atmost2}, let $\{x_1,x_2\} \supseteq N_{G[L]}(v)\cap T$. Let $\{i_1,i_2\}=\{1,2,3\}\setminus \{i\}$. Now, we set 
        $R_i:=R_i\cup \{x_1,x_2\}$, $R_{i_1}:=R_{i_1}\cup \{x_2,x_3\}$, and $R_{i_2}:=R_{i_2}\cup \{x_3,x_1\}$. 
    \end{itemize}
\end{itemize}
This completes the construction of $R_1,R_2$, and $R_3$.  
For each $i\in \{1,2,3\}$, let $R_i'=R_i\cap L$. 
From the construction of $R_1,R_2, R_3, R_1',R_2'$, and $R'_{3}$,  
the following properties are satisfied.   

\begin{itemize}
     \item[(i)] For each triangle $T$ in ${\cal L}$, exactly two vertices from $T$ belong to $R_i$ for any $i\in \{1,2,3\}$. This implies that $|R_i'|=2\vert {\cal L}\vert$ for all $i\in \{1,2,3\}$. 
     \item[(ii)] For each triangle $T$ in ${\cal L}$ and $x\in T$, there are exactly two distinct indices $i,j\in \{1,2,3\}$ such that $x\in R_i$ and $x\in R_j$ 
     \item[(iii)] For each $i\in \{1,2,3\}$, $R_i\cap I_i=\emptyset$ and $N_{G[L]}(I_i)\subseteq R_i$.  
     \item[(iv)] For each triangle $T$ in ${\cal H}$ and $i\in \{1,2,3\}$, $T\subseteq R_i$.  
 \end{itemize}

Recall that $S_{\sf opt}$ is an optimum solution. Next, we prove a lemma that is analogous to Observation~\ref{lem:largehitopt}. 

\begin{lemma}
There is an index $j\in \{1,2,3\}$ such that $|R_j\cap S_{\sf opt}|\geq \frac{\vert R_j \vert}{3}$. 
\end{lemma}
\begin{proof}
Notice that $S_{\sf opt}$ contains at least one vertex from each triangle $T$ in ${\cal L}$. This statement along with the property (ii) implies the inequality
\begin{equation*}
|R'_1\cap S_{\sf opt}|+|R'_2\cap S_{\sf opt}|+|R'_3\cap S_{\sf opt}| \geq 2 |{\cal L}|.
\end{equation*}
This implies that there is an index $j\in \{1,2,3\}$ such that 
$|R'_j\cap S_{\sf opt}| \geq \frac{2\vert {\cal L}\vert}{3}$. By property (ii) above, we know that $|R_j'|=2\vert {\cal L}\vert$. Thus, we get that $|R'_j\cap S_{\sf opt}| \geq \frac{|R_j'|}{3}$. 

By property (iv) above, we get that $R_j\setminus R_j'=V({\cal H})$ and $\vert  R_j\setminus R_j'\vert = 3\vert {\cal H}\vert$. For each triangle $T$ in ${\cal H}$, there is at least one vertex from $T$ that belongs to $S_{\sf opt}$. Thus, we get
\begin{equation*}
|(R_j\setminus R_j')\cap S_{\sf opt}| \geq |\mathcal{H}| = \frac{|R_j \backslash R_j'|}{3}.
\end{equation*}
Finally, we have $|R_j\cap S_{\sf opt}| \geq \frac{|R_j|}{3}$,
as $|R'_j\cap S_{\sf opt}| \geq \frac{|R_j'|}{3}$ and $|(R_j\setminus R_j')\cap S_{\sf opt}| \geq \frac{\vert R_j\setminus R_j'\vert}{3}$.
\end{proof}

Now recall the definition of a dead vertex. We say that a vertex $v\in V({\cal T})$ is {\em dead} with respect to $R_i$, where $i\in \{1,2,3\}$,  if $v\notin R_i$ and $v$ is not contained in any triangle in $G[V({\cal T})\setminus R_i]$. Next we prove that there are $\Omega(|{\cal T}|)$ many dead vertices with respect to $R_i$ for each $i\in \{1,2,3\}$. Towards that we prove that for each $i\in \{1,2,3\}$, all the vertices in $I_i$ are dead with respect to $R_i$ and $\vert I_i \vert = \Omega(|{\cal T}|)$. 

\begin{lemma}
For any  $i\in \{1,2,3\}$, all the vertices in $I_i$ are dead with respect to $R_i$
\end{lemma}

\begin{proof}
Fix an index $i\in \{1,2,3\}$. From properties (iii) and (iv) we get that $I_i \subseteq V({\cal T}) \setminus R_i$ and each vertex in $I_i$ is an isolated vertex in $G[V({\cal T})\setminus R_i]$. Therefore, all the vertices  in $I_i$ are dead with respect to $R_i$. 
\end{proof}

\begin{lemma}
For any $i\in \{1,2,3\}$, $\vert I_i \vert \geq \frac{\vert {\cal T}\vert}{400^3}$
\end{lemma}

\begin{proof}
First, we prove that $\vert I \vert \geq \frac{3|{\cal T}|}{4^2 \cdot 100^3} $. 
Since the degree of each vertex in $G[L]$ is at most $100$,  $N_{G[L]}[N_{G[L]}[N_{G[L]}[I]]] \leq |I|+100 |I|+ 100^2 |I|+100^3|I| \leq 4 \cdot 100^3 |I|$. Since $|I|$ is a maximal distance-$3$ independent set in $G[L]$, $N_{G[L]}[N_{G[L]}[N_{G[L]}[I]]]=L$. This implies that $|I|\geq \frac{|L|}{4 \cdot 100^3}$.  

Let $n'=|V({\cal T})|=3|{\cal T}|$. 
By Lemma~\ref{lem-degen}, we know that $G$ is $11$-degenerate. Hence $G[V({\cal T})]$ is $11$-degenerate. This implies 
that the average degree of a vertex in $G[V({\cal T})]$ is at most $22$.  
Thus, the number of vertices with degree at least 100 is at most $\frac{n'}{4}$. This implies that $|H|\leq \frac{3n'}{4}$ and $|L|\geq \frac{n'}{4}=\frac{3|{\cal T}|}{4}$. Thus, since $|I|\geq \frac{|L|}{4 \cdot 100^3}$, we get that $|I|\geq \frac{3|{\cal T}|}{4^2 \cdot 100^3}$. 

Now, for any $i\in \{1,2,3\}$, we know that $|I_i|=\lfloor \frac{|I|}{3}\rfloor \geq \frac{|I|}{3}-2$. Therefore, we get that  for any $i\in \{1,2,3\}$, $|I_i|\geq \frac{|{\cal T}|}{400^3}$. This completes the proof of the lemma. 
\end{proof}

To summarize, there exists $i \in \{1,2,3\}$ such that $|R_i \cap S_\mathsf{opt}| \geq \frac{|R_i|}{3}$ and the number of dead vertices in $V(\mathcal{T})$ with respect to $R_i$ is at least $\frac{|\mathcal{T}|}{400^3}$.
Without loss of generality, we can assume $R_1$ satisfies this condition (in the algorithm, we do not know which one is good, but we can try all three and take the best one among the three resulting solutions).
Set $R = R_1$ and let $D \subseteq V(\mathcal{T})$ be the set of dead vertices with respect to $R$.
Now Observations~\ref{lem:largehitopt} and~\ref{obs-T'dead} still hold.
Therefore, we still have Equation~\eqref{eq-S2D} (without the expectation notations $\mathbb{E}[ \cdot ]$), i.e.,
\begin{equation*}
    |S_2| \leq (\rho_0+(3-\rho_0)(a+b)) \cdot \mathsf{opt} - (3-\rho_0) \cdot \frac{|D|}{3}.
\end{equation*}
We have $|D| \geq \frac{|\mathcal{T}|}{400^3} = \frac{a \cdot \mathsf{opt}}{400^3}$, so the above inequality implies
\begin{equation*}
    |S_2| \leq (\rho_0+(1-400^{-3})(3-\rho_0)a + (3-\rho_0)b) \cdot \mathsf{opt}.
\end{equation*}
Now the only difference between the above inequality and the one in Lemma~\ref{lem-S2} is that the number $(\frac{1}{3})^{\frac{3}{8}d + \frac{1}{4}}$ is replaced with an even smaller constant $400^{-3}$.
But as long as this number is a positive constant, our analysis in Section~\ref{sec-analysis} can give us an approximation ratio better than 3.
Also, we can directly see from the inequality that if $b$ is sufficiently small, no matter what $a$ is, we always have $|S_2| \leq (3-\alpha) \cdot \mathsf{opt}$ for some constant $\alpha>0$.

\end{document}